\documentclass{tlp}

\usepackage{tikz}
\usepgflibrary{arrows}

\usepackage{amssymb}
\usepackage{latexsym}
\usepackage{url}
\usepackage{times}

\newtheorem{definition}{Definition}
\newtheorem{theorem}{Theorem}
\newtheorem{proposition}{Proposition}
\newtheorem{lemma}{Lemma}

\newcommand{\U}{{\cal U}}
\newcommand{\AS}{{\mathit AS}}
\newcommand{\At}{{\mathit At}}
\newcommand{\la}{\leftarrow}
\newcommand{\naf}{{\mathit not}}

\newcommand{\q}{\ensuremath{\mathit{q}}}
\newcommand{\qnew}{\ensuremath{\mathit{q'}}}
\newcommand{\fail}{\ensuremath{\mathit{fail}}}
\newcommand{\strat}[1]{\ensuremath{#1_\mathit{strat}}}

\renewcommand{\P}{\mathcal{P}}
\renewcommand{\P}{P} 
\newcommand{\p}{\P}
\newcommand{\R}{\ensuremath{r}}
\newcommand{\HR}{\ensuremath{H(\R)}}
\newcommand{\BR}{\ensuremath{B(\R)}}
\newcommand{\head}[1]{\ensuremath{H(#1)}}
\newcommand{\body}[1]{\ensuremath{B(#1)}}
\newcommand{\posbody}[1]{\ensuremath{B^+(#1)}}
\newcommand{\negbody}[1]{\ensuremath{B^-(#1)}}

\newcommand{\NP}{\ensuremath{\mathrm{NP}}}
\newcommand{\NEXP}{\ensuremath{\mathrm{NEXP}}}

\newcommand{\CONEXP}[1]{\ensuremath{\mbox{\rm co-}\NEXP^{#1}}}
\newcommand{\SigmaP}[1]{\ensuremath{\Sigma_{#1}^P}}
\newcommand{\PiP}[1]{\ensuremath{{\Pi}_{#1}^{P}}}

\newcommand{\I}{{\mathcal I}}
\newcommand{\M}{{\mathcal M}}

\newcommand{\ASP}{\ensuremath{\rm ASP}}
\newcommand{\ASPSC}{\ensuremath{\rm ASP^{\rm sc}}}

\newcommand{\citeText}[1]{\citeANP{#1} \citeNN{#1}}

\begin{document}

\title[Complexity of super-coherence problems in ASP]{%
Complexity of Super-Coherence Problems in ASP
\thanks{%
Preliminary versions of this article have been presented at 
at the ICLP workshop 
on Answer Set Programming and Other Computing Paradigms (ASPOCP)
and at the 
Convegno Italiano
di Logica Computazionale (CILC).
This work is 
partly supported by Regione Calabria and EU
under POR Calabria FESR 2007-2013 and within the PIA
project of DLVSYSTEM s.r.l., by MIUR under the PRIN project LoDeN, 
and by the Vienna University of Technology under the program ``Innovative Ideas''.
We also thank the anonymous reviewers from 
ASPOCP and CILC
for their valuable comments.
}
}

\author[M. Alviano, W. Faber and S. Woltran]{%
        MARIO ALVIANO and WOLFGANG FABER\\
        University of Calabria, 87036 Rende (CS), Italy\\
	\email{\{alviano,faber\}@mat.unical.it}
        \and STEFAN WOLTRAN\\
	Vienna University of Technology, 1040 Vienna, Austria\\
	\email{woltran@dbai.tuwien.ac.at}}

\submitted{14 Decempber 2011}
\revised{8 October 2012}
\accepted{19 December 2012}

\maketitle

\label{firstpage}

\begin{abstract}
Adapting techniques from database theory in order to optimize Answer
Set Programming (ASP) systems, and in particular the grounding
components of ASP systems, is an important topic in ASP. 
In recent years, the Magic Set method has received some interest in this
setting, and a variant of it, called DMS, has been proposed for ASP.
However, this technique has a caveat,
because it is not correct (in the sense of being query-equivalent) for
all ASP programs. In recent work,
a large fragment of
ASP programs, referred to as \emph{super-coherent programs}, has been
identified, for which DMS is correct. 
The fragment contains all
programs which
possess at least one answer set, no matter which set of facts is added
to them. 
Two open question remained: How
complex is it to determine whether a given program is super-coherent?
Does the restriction to super-coherent programs limit the problems
that can be solved?
Especially the first question turned out to be quite difficult to answer precisely.
In this paper, we formally prove that deciding whether a propositional
program is super-coherent is $\PiP{3}$-complete in the disjunctive
case, while it is $\PiP{2}$-complete for normal programs. The hardness
proofs are the difficult part in this endeavor: We proceed by
characterizing the reductions by the models and reduct models which
the ASP programs should have, and then provide instantiations
that meet the given specifications. 
Concerning the second question, we show that all relevant ASP
reasoning tasks can be transformed into tasks over super-coherent
programs, even though this transformation is more of theoretical than
practical interest.
\end{abstract}

\begin{keywords}
Answer-Set Programming, Complexity Analysis
\end{keywords}

\section{Introduction}\label{sec:introduction}

Answer Set Programming (\ASP) is a powerful formalism for knowledge
representation and common sense reasoning \cite{bara-2002}.  Allowing
disjunction in rule heads and nonmonotonic negation in bodies,
\ASP\ can express every query belonging to the complexity class
$\rm\SigmaP2$ ($\NP^\NP$).  Encouraged by the availability of
efficient inference engines, such as DLV~\cite{leon-etal-2002-dlv},
GnT~\cite{janh-etal-2005-tocl}, Cmodels~\cite{lier-2005-lpnmr}, or
ClaspD~\cite{dres-etal-2008-KR}, \ASP\ has found several practical
applications in various domains, including data integration \cite{leon-etal-2005}, semantic-based information extraction \cite{manna-etal-2011-tldks,manna-etal-2011-jcss}, e-tourism \cite{ricca-etal-2010-IDUM}, workforce management \cite{ricca-etal-2012-tplp}, and many more. As a matter of fact, these
\ASP\ systems are continuously enhanced to support novel optimization
strategies, enabling them to be effective over increasingly larger
application domains.

Frequently, optimization techniques are inspired by methods that had
been proposed in other fields, for example database theory,
satisfiability solving, or constraint satisfaction. Among techniques
adapted to ASP from database theory, Magic Sets
\cite{ullm-89,banc-etal-86,beer-rama-91} have recently achieved a lot
of attention. Following some earlier work
\cite{grec-2003,cumb-etal-2004-iclp}, an adapted method
called \emph{DMS} has been proposed for \ASP\ in
\cite{alvi-etal-2012-aij}. However, this technique has a caveat,
because it is not correct (in the sense of being query-equivalent) for
all \ASP\ programs. In recent work
\cite{alvi-fabe-2011-aicom,alvi-fabe-2010-aspocp}, a large fragment of
ASP programs, referred to as \emph{super-coherent programs} (\ASPSC),
has been identified, for which DMS can be proved to be correct.
Formally, a program is super-coherent, if 
it is coherent (i.e.\ possesses at least one answer set), 
no matter which input (given as a set of facts) is added to the program.

Since the property of being super-coherent is a semantic one, a
natural question arises: How difficult is it to decide whether a given
program belongs to \ASPSC? It turns out that the precise complexity is
rather difficult to establish. Some bounds have been given in
\cite{alvi-fabe-2011-aicom}, in particular showing decidability, but
especially hardness results seemed quite hard to obtain.
In particular, the following question remained unanswered:
Is it possible to implement an efficient algorithm for testing
super-coherence of a program, to decide for example whether DMS
has to be applied or not? 
In this paper we provide a negative answer to this question, proving
that deciding whether a propositional program is super-coherent is
complete for the third level of the polynomial hierarchy in the general
case, and for the second level for normal programs.
As the complexity of query answering 
is located on lower levels 
of the polynomial hierarchy, our results show that implementing a sound and complete 
super-coherence check in a query optimization setting does not provide an approach
for improving such systems.

While our main motivation for studying \ASPSC\ stemmed from the
applicability of DMS, this class actually has many more important
motivations. Indeed, it can be viewed as the class of
\emph{non-constraining programs}: Adding extensional information to
these programs will always result in answer sets.  One important
implication of this property is for modular evaluation. For instance,
when using the splitting set theorem of \citeText{lifs-turn-94}, if a top
part of a split program is an \ASPSC\ program, then any answer set of
the bottom part will give rise to at least one answer set of the full
program---so for determining answer set existence, there would be no
need to evaluate the top part.

On a more abstract level, one of
the main criticisms of \ASP\ (being voiced especially in database
theory) is that there are programs which do not admit any answer set
(traditionally this has been considered a more serious problem than
the related nondeterminism in the form of multiple answer sets, cf.~\citeNP{papa-yann-97}). From
this perspective, programs which guarantee coherence (existence of an
answer set) have been of interest for quite some time. In particular,
if one considers a fixed program and a variable ``database,'' one
arrives naturally at the class \ASPSC\ when requiring existence of an
answer set. This also indicates that deciding super-coherence of programs 
is related to some problems from the area of equivalence checking in ASP~\cite{Eiter05,Eiter07,Oetsch07b}. For instance, when deciding 
whether, for a given arbitrary program $P$, there is a uniformly equivalent
definite positive (or definite Horn) program, super-coherence of $P$ is a 
necessary condition---this is straightforward to see because definite Horn
programs have exactly one answer set, so a non-super-coherent program cannot
be uniformly equivalent to any definite Horn program.

Since super-coherent programs form a strict subset of all \ASP{}
programs, another important question arises: Does the restriction to
super-coherent programs limit the problems that can be solved by them?
In this paper, we show that this is not the case, by embedding all
relevant reasoning tasks over \ASP{} (testing answer set existence,
query answering, answer set computation) into reasoning tasks over
\ASPSC{}. We also show that all reasoning tasks over normal
(non-disjunctive) \ASP{} can be embedded into tasks over normal
\ASPSC{}. These results essentially demonstrate that \ASPSC{} is
sufficient to encode any problem that can be solved by full \ASP{},
and is therefore in a sense ``complete''. However, we would like to
note that these embeddings were designed for answering this
theoretical question, and might lead to significant overhead when
evaluated with ASP solvers. We therefore do not advocate to use them
in practical settings, and finding efficient embeddings is a
challenging topic for future research.

To summarize, the main contributions of the paper are as follows:
\begin{itemize}
\item We show that 
recognizing 
super-coherence for disjunctive and normal programs is 
complete for classes 
$\PiP{3}$ and $\PiP{2}$, respectively, thus more complex than the common ASP reasoning tasks.  
\item 
We provide a transformation of 
reasoning tasks over general programs into tasks over
super-coherent programs, showing that the restriction to super-coherent programs does not curtail expressive power.
\item We also briefly discuss the relation between checking for super-coherence
and testing equivalence between programs where we make use of our results to sharpen complexity
results 
due to \citeText{Oetsch07b}.
\end{itemize}

In order to focus on the essentials of these problems, in this paper
we deal with propositional programs, but we conjecture that the results can be extended
to the non-propositional case by using complexity upgrading techniques as presented in \cite{eite-etal-97-tods,gott-etal-99apal}, arriving at completeness for classes \CONEXP{\NP} and \CONEXP{\SigmaP{2}}, respectively. 

The remainder of this article is organized as follows.
In
Section~\ref{sec:preliminaries} we first define some terminology
needed later on. In Section~\ref{sec:statement} we formulate
the complexity problems that we analyze, and state our main results. The
proofs for these problems are presented in
Section~\ref{sec:proofdisj} for disjunctive programs, and in
Section~\ref{sec:proofnorm} for normal programs.
In Section~\ref{sec:embedding} we show
``completeness'' of \ASPSC{} via simulating reasoning tasks over \ASP{}
by tasks over \ASPSC{}.
In Section~\ref{sec:impl} we briefly discuss the relation to equivalence
problems before concluding the work in Section~\ref{sec:concl}.

\section{Preliminaries}\label{sec:preliminaries}

In this paper we consider propositional programs, so an atom $p$ is a
member of a countable set $\U$. A {\em literal} is either an atom
$p$ (a positive literal), or an atom preceded by the {\em
  negation as failure} symbol $\naf$ (a negative literal).
A {\em rule} $\R$ is of the form
\[
\begin{array}{l}
    p_1 \ \lor\ \cdots \ \lor\ p_n \la
    q_1,\ \ldots,\ q_j,\ \naf~q_{j+1},\ \ldots,\ \naf~q_m  
\end{array}
\]
where $p_1,\ \ldots,\ p_n,\ q_1,\ \ldots,\ q_m$ 
are atoms and $n\geq 0,$  $m\geq j\geq 0$. The
disjunction $p_1\ \lor\ \cdots \ \lor\ p_n$ is the {\em head} of~\R{}, 
while the conjunction 
$q_1,\ \ldots,\ q_j,\ \naf~q_{j+1},\ \ldots,\ \naf~q_m$ 
is the {\em body} of~\R{}.
Moreover, $\HR$ denotes the set of head atoms, while $\BR$ denotes the set of body literals.
We also use $\posbody{\R}$ and $\negbody{\R}$ for denoting
the set of atoms appearing in positive and negative body literals, respectively,
and $\At(\R)$ for the set $\HR \cup \posbody{\R} \cup \negbody{\R}$.
A rule $\R$ is normal (or disjunction-free) if $|\HR| = 1$ or $|\HR| = 0$ 
(in this case $\R$ is also referred to as a \emph{constraint}), 
positive (or negation-free) if $\negbody{\R}=\emptyset$,
a {\em fact} if both $\body{\R}=\emptyset$ and
$|\HR| = 1$.

A \emph{program}
$\P$ is a finite set of rules; if all rules in it are positive (resp.\ normal), 
then $\P$ is a positive (resp.\ normal) program.
Odd-cycle-free (cf.\ \cite{dung-92,you-yuan-94}) and stratified (cf.\ \cite{apt-etal-88}) programs constitute two other interesting classes of 
programs. An atom $p$ appearing in the head of a rule $\R$
{\em depends} on each atom $q$ that belongs to $\BR$; 
if $q$ belongs to $\posbody{\R}$, $p$ depends positively on $q$, otherwise negatively. 
A program without constraints is \emph{odd-cycle-free} if there is no cycle of dependencies involving an
odd number of negative dependencies, while it is {\em stratified} if each
cycle of dependencies involves only positive dependencies.
Programs containing constraints have been excluded by the definition of odd-cycle-free
and stratified programs. In fact, constraints intrinsically introduce odd-cycles 
in programs as a constraint of the form
$$\la q_1,\ \ldots,\ q_j,\ \naf~q_{j+1},\ \ldots,\ \naf~q_m$$
can be replaced by the following equivalent rule:
$$co \la q_1,\ \ldots,\ q_j,\ \naf~q_{j+1},\ \ldots,\ \naf~q_m,\ \naf~co,$$
where $co$ is a fresh atom (i.e., an atom that does not occur elsewhere in the program). We also require the notion of \emph{head-cycle free} programs (cf.\ \cite{bene-dech-94}): a 
program $\P$ is called head-cycle free 
if no different head atoms in a rule positively depend mutually on
each other.

Given a program $\P$, let $\At(\P)$ denote the set of atoms that occur in it,
that is, let $\At(\P) = \bigcup_{\R\in \P} \At(\R)$.
An \emph{interpretation} $I$ for a program $\P$ is a subset of $\At(\P)$. An
atom $p$ is true w.r.t.\ an
interpretation $I$ if $p\in I$; otherwise, it is false. 
A negative literal $\naf\ p$ is true w.r.t.\ $I$ 
if and only if $p$ is false w.r.t.\ $I$. 
The body of a rule $\R$ is true w.r.t.\ $I$
if and only if all the body literals of $\R$ are true w.r.t.\ $I$, that is,
if and only if $\posbody{\R} \subseteq I$ and $\negbody{\R} \cap I = \emptyset$. 
An interpretation $I$ {\em satisfies} a rule $\R\in \P$ if at least one atom
in $\head{\R}$ is true w.r.t.\ $I$ whenever the body of $\R$ is true w.r.t.\ $I$. An interpretation $I$ is a
\emph{model} of a program $\P$ if $I$ satisfies all the rules in $\P$. 

Given an interpretation $I$ for a program $\P$, the reduct of $\P$ w.r.t.\ $I$,
denoted by $\P^{I}$, is obtained by deleting from $\P$ all 
the rules $\R$ with $\negbody{\R} \cap I \neq \emptyset$,
and then by removing all the negative literals from the remaining rules~\cite{gelf-lifs-91}.
The semantics of a program $\P$ is given by the set $\AS(\P)$ of the answer sets of
$\P$, where an interpretation $M$ is an answer set for $\P$ if and only if
$M$ is a subset-minimal model of $\P^M$.

In the subsequent sections, we will use the following properties that the models and 
models of reducts of programs satisfy (see, e.g. \cite{Eiter04,Eiter05}):
\begin{itemize}
\item[(P1)] for any disjunctive program $P$ and interpretations $I\subseteq J\subseteq K$, if $I$ satisfies $P^J$, then $I$ also satisfies $P^K$;
\item[(P2)] for any normal program $P$ and interpretations $I, J\subseteq K$,
if $I$ and $J$ both satisfy $P^K$, then also $(I\cap J)$ satisfies $P^K$. 
\end{itemize}

By a query in this paper we refer to an atom, negative and conjunctive queries can be simulated by adding appropriate rules to the associated program.
A query $\q$ is \emph{bravely} true for a program $P$ if and only if $q\in A$ for some $A \in \AS(P)$.
A query $\q$ is \emph{cautiously} true for a program $P$ if and only if $q\in A$ for all $A \in \AS(P)$.

We now introduce super-coherent ASP programs (\ASPSC{} programs), 
the main class of programs studied in this paper.

\begin{definition}[\ASPSC{} programs;~\citeNP{alvi-fabe-2010-aspocp,alvi-fabe-2011-aicom}]
A program $\P$ is \emph{super-coherent} if, for every set of facts $F$,
 $\AS(\P \cup F) \neq \emptyset$. Let \ASPSC{} denote the set of all super-coherent programs.
\end{definition}

Note that \ASPSC{} programs include all odd-cycle-free programs
(and therefore also all stratified programs). 
Indeed, every odd-cycle-free program admits at least one answer set
and remains odd-cycle-free even if an arbitrary set of facts is added
to its rules. On the other hand,
there are programs having odd-cycles that are in \ASPSC{}, cf.~\citeText{alvi-fabe-2011-aicom}.

\section{Complexity of Checking Super-Coherence}\label{sec:statement}

In this section, we study the complexity of the following natural problem:
\begin{itemize}
\item 
Given a program $P$, 
is $P$ super-coherent, i.e.\ 
does $\AS(P\cup F)\neq\emptyset$ hold for any
set $F$ of facts.
\end{itemize}

We will study the complexity for this problem for the 
case of disjunctive logic programs and non-disjunctive (normal)
logic programs.
%
We first have a look at a similar problem, which turns out to be
rather trivial to decide.

\begin{proposition}
The problem of deciding whether, for a given disjunctive program $P$, there
is a set $F$ of facts such that $\AS(P\cup F)\neq\emptyset$ is $\NP$-complete; $\NP$-hardness holds already for normal programs.
\end{proposition}
\begin{proof}
We start by observing that there is $F$ such that $\AS(P\cup F)\neq\emptyset$
if and only if $P$ has at least one
model.  Indeed, if $M$ is a model of $P$, then $P\cup M$ has $M$ 
as its answer set. On the other hand, if $P$ has no model, then no addition of facts $F$ will
yield an answer set for $P\cup F$. It is well known that deciding 
whether a program has at least one (classical) model is $\NP$-complete for both
disjunctive and normal programs.
\end{proof}

In contrast, the complexity for deciding 
super-coherence is surprisingly high, which
we shall show next. To start, we give a straight-forward observation.

\begin{proposition}\label{prop:restrict}
A program $P$ is super-coherent if and only if 
for each set $F\subseteq \At(P)$,  $\AS(P\cup F)\neq\emptyset$.
\end{proposition}
\begin{proof}
The only-if direction is by definition. For the if-direction, let $F$ be any 
set of facts. $F$ can be partitioned into $F' = F \cap \At(P)$ and
$F'' = F \setminus F'$. By assumption, $P$ is super-coherent and thus
$P \cup F'$ is coherent.
Let $M$ be an answer set of $P \cup F'$.
We shall show that $M \cup F''$ is an answer set of $P \cup F = P \cup F' \cup F''$.
This is in fact a consequence of the splitting set theorem~\cite{lifs-turn-94},
as the atoms in $F''$ are only defined by facts not occurring in $P \cup F'$.
\end{proof}

\noindent
Our main results are stated below. The proofs are contained in the subsequent
sections.

\begin{theorem}\label{thm:disj}
The problem of deciding super-coherence for disjunctive programs 
is $\PiP{3}$-complete.
\end{theorem}

\begin{theorem}\label{thm:norm}
The problem of deciding super-coherence for normal programs 
is $\PiP{2}$-complete.
\end{theorem}

\subsection{Proof of Theorem~\ref{thm:disj}}\label{sec:proofdisj}

Membership follows by the following straight-forward nondeterministic 
algorithm for the complementary problem, 
i.e.\ given a program $P$, does there exist a set $F$ of facts such that
$\AS(P\cup F)=\emptyset$: we
guess a set $F\subseteq \At(P)$ and check $\AS(P\cup F)=\emptyset$ via an 
oracle-call.
Restricting the guess to $\At(P)$ can be done by Proposition~\ref{prop:restrict}. Checking $\AS(P\cup F)=\emptyset$ is known to be in $\PiP{2}$ \cite{EiterG95}.
This shows $\PiP{3}$-membership.

For the hardness we reduce the $\PiP{3}$-complete problem of deciding 
whether QBFs of the form 
$\forall X\exists Y\forall Z\phi$ are true to the problem of super-coherence.
Without loss of generality, we can consider $\phi$ to be in DNF and, indeed, 
$X\neq\emptyset$, $Y\neq \emptyset$, and $Z\neq\emptyset$.
We also assume that each disjunct of $\phi$ contains at least one variable from $X$,
one from $Y$ and one from $Z$.
More precisely, we shall construct for each such QBF $\Phi$ a program 
$P_\Phi$ such that $\Phi$ is true if and only if $P_\Phi$ is super-coherent.
Before showing how to actually construct $P_\Phi$ from $\Phi$ in polynomial time, we give the required properties for $P_\Phi$. We then show that 
for programs $P_\Phi$ satisfying these properties, the desired relation 
($\Phi$ is true if and only if $P_\Phi$ is super-coherent) holds, and finally 
we provide the construction of $P_\Phi$.

\begin{definition}\label{def:red}
Let $\Phi = \forall X\exists Y\forall Z\phi$ be a QBF with $\phi$ in DNF.
We call any program $P$ satisfying the following properties 
a \emph{$\Phi$-reduction}:
\begin{enumerate}
\item $P$ is given over atoms $U=X\cup Y\cup Z \cup 
\overline{X} \cup 
\overline{Y} \cup 
\overline{Z} \cup  \{u,v,w\}$, where all atoms in sets 
$\overline{S}=\{\overline{s} \mid s\in S\}$ ($S \in \{X,Y,Z\}$) and $\{u,v,w\}$ are fresh and mutually disjoint;
\item $P$ has the following models: 
\begin{itemize}
\item 
$U$;
\item 
for each 
$I\subseteq X$,
$J\subseteq Y$, 
$$
M[I,J]=I\cup \overline{(X\setminus I)}\cup J\cup \overline{(Y\setminus J)}\cup Z\cup \overline{Z}\cup \{u,v\}
$$
and
$$
M'[I,J]=I\cup \overline{(X\setminus I)}\cup J\cup \overline{(Y\setminus J)}\cup Z\cup \overline{Z}\cup \{v,w\};
$$
\end{itemize}
\item for each $I\subseteq X$, $J\subseteq Y$, the models%
\footnote{Here and below, for a reduct $P^M$ we only list models of the form $N\subseteq M$, since
those are the relevant ones for our purposes. Recall that $N=M$ is always a model of $P^M$ in case $M$ is a model of $P$.}
 of the reduct
$P^{M[I,J]}$ are $M[I,J]$ and
$$
O[I] = I\cup \overline{(X\setminus I)};
$$
\item for each $I\subseteq X$, $J\subseteq Y$, the models of the reduct
$P^{M'[I,J]}$ are $M'[I,J]$ and
\begin{itemize}
\item 
for each $K\subseteq Z$ such that $I\cup J \cup K\not \models \phi$,
$$
N[I,J,K] = I \cup \overline{(X\setminus I)}\cup J\cup \overline{(Y\setminus J)}\cup K\cup \overline{(Z\setminus K)}\cup \{v\};
$$
\end{itemize}
\item  the models of the reduct $P^U$ are $U$ itself, 
$M[I,J]$, $M'[I,J]$, and $O[I]$, for each $I\subseteq X$, $J\subseteq Y$, and
$N[I,J,K]$ for each $I\subseteq X$, $J\subseteq Y$, $K\subseteq Z$, such that $I\cup J \cup K\not \models \phi$.
\end{enumerate}

\end{definition}

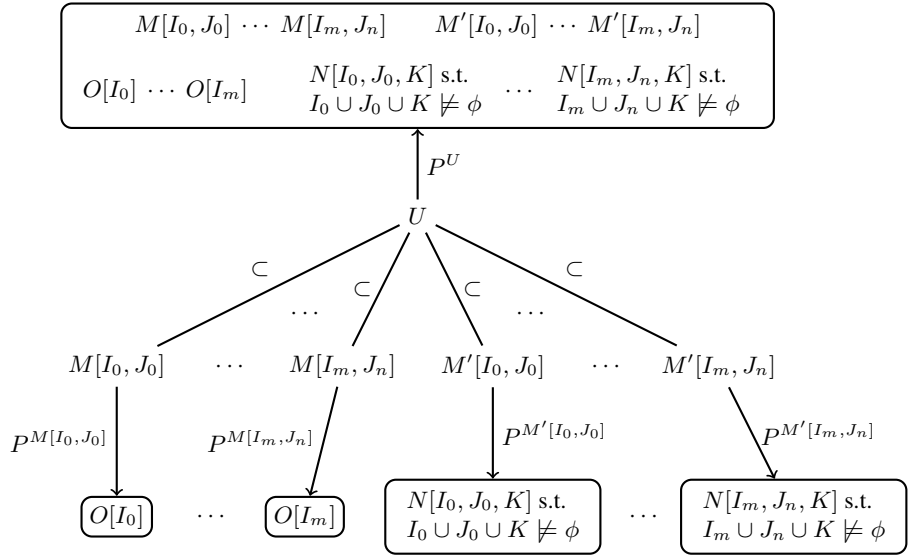
\begin{figure}
  
\tikzstyle{vertex}=[]
\tikzstyle{arc}=[draw,thick,->]
\tikzstyle{edge}=[draw,thick]
\tikzstyle{countermodels}=[draw,thick,rectangle, rounded corners]

\noindent
\begin{tikzpicture}
\node[vertex] (U) at (0,1) {$U$};
\node[vertex] (MI0J0) at (-4,-1) {$M[I_0,J_0]$};
\node[vertex] (dots1) at (-2.5,-1) {$\cdots$};
\node[vertex] (MImJn) at (-1,-1) {$M[I_m,J_n]$};
\node[vertex] (MpI0J0) at (1,-1) {$M'[I_0,J_0]$};
\node[vertex] (dots2) at (2.5,-1) {$\cdots$};
\node[vertex] (MpImJn) at (4,-1) {$M'[I_m,J_n]$};

\path[edge] (U) edge node[above left] {$\subset$} (MI0J0);
\node[vertex] (dots3) at (-1.5,-0.3) {$\cdots$};
\path[edge] (U) edge node[left] {$\subset$} (MImJn);

\path[edge] (U) edge node[right] {$\subset$} (MpI0J0);
\node[vertex] (dots4) at (1.5,-0.3) {$\cdots$};
\path[edge] (U) edge node[above right] {$\subset$} (MpImJn);


\node[countermodels] (CMMI0J0) at (-4,-3) {$O[I_0]$};
\path[arc] (MI0J0) edge node[left] {$P^{M[I_0,J_0]}$} (CMMI0J0);

\node[vertex] (dots5) at (-2.75,-3) {$\cdots$};

\node[countermodels] (CMMImJn) at (-1.5,-3) {$O[I_m]$};
\path[arc] (MImJn) edge node[left] {$P^{M[I_m,J_n]}$} (CMMImJn);

\node[countermodels] (CMMpI0J0) at (1,-3) {$
  \begin{array}{l}
   N[I_0,J_0,K] \mbox{\ s.t.}\\
   I_0 \cup J_0 \cup K \not \models \phi  
  \end{array}
$};
\path[arc] (MpI0J0) edge node[right] {$P^{M'[I_0,J_0]}$} (CMMpI0J0);

\node[vertex] (dots5) at (3,-3) {$\cdots$};

\node[countermodels] (CMMpImJn) at (5,-3) {$
  \begin{array}{l}
   N[I_m,J_n,K] \mbox{\ s.t.}\\
   I_m \cup J_n \cup K \not \models \phi  
  \end{array}
 $};
\path[arc] (MpImJn) edge node[right] {$P^{M'[I_m,J_n]}$} (CMMpImJn);

\node[countermodels] (CMU) at (0,3) {$
  \begin{array}{c}
   M[I_0,J_0]\ \cdots\ M[I_m,J_n]\qquad M'[I_0,J_0]\ \cdots\ M'[I_m,J_n]\\[2ex]  
   O[I_0]\ \cdots\ O[I_m] \qquad
   \begin{array}{l}
   N[I_0,J_0,K] \mbox{\ s.t.}\\
   I_0 \cup J_0 \cup K \not \models \phi  
  \end{array}
\ \cdots\ 
  \begin{array}{l}
   N[I_m,J_n,K] \mbox{\ s.t.}\\
   I_m \cup J_n \cup K \not \models \phi  
  \end{array}
  \end{array}
$};

\path[arc] (U) edge node[right] {$P^{U}$} (CMU);

\end{tikzpicture}

  \caption{Models and reduct ``countermodels'' of $\Phi$-reductions}
  \label{fig:phired}
\end{figure}

The structure of models of $\Phi$-reductions and the ``countermodels''
(see below what we mean by this term) of the relevant reducts is sketched in Figure~\ref{fig:phired}. The
center of the diagram contains the models of the $\Phi$-reduction and
their subset relationship. For each of the model the respective box
lists the ``countermodels,'' by which we mean those reduct models which can serve as counterexamples for the original model being an answer set, that is, those reduct models which are proper subsets of the original model.

We just note at this point that the models of the reduct $P^U$ 
given in Item~5 are not specified for particular purposes, but 
are required to allow for a realization via disjunctive programs. 
In fact, these models are just an effect of property (P1) mentioned
in Section~\ref{sec:preliminaries}. However, before 
showing a program satisfying the 
properties of
a \emph{$\Phi$-reduction}, we first show the rationale behind
the concept of $\Phi$-reductions.

\begin{lemma}
For any QBF $\Phi = \forall X\exists Y\forall Z\phi$ with $\phi$ in DNF,
a $\Phi$-reduction is super-coherent if and only if 
$\Phi$ is true.
\end{lemma}
\begin{proof}
Suppose that $\Phi$ is false. 
Hence, there exists an $\I\subseteq X$ such that, for all 
$J\subseteq Y$, there is a $K_Y\subseteq Z$ with 
$\I\cup J \cup K_Y\not\models \phi$.
Now let $P$ be any $\Phi$-reduction and 
$F_\I = \I\cup \overline{(X\setminus \I)}$. 
We show that $\AS(P\cup F_\I)=\emptyset$, 
thus $P$ is not super-coherent.
Let $\M$ be a model of $P \cup F_\I$.
Since $P$ is a $\Phi$-reduction, the 
only candidates for $\M$ are 
$U$, $M[\I,J]$, and $M'[\I,J]$, 
where $J\subseteq Y$.  Indeed, for each $I\neq\I$, 
$M[I,J]$ and $M'[I,J]$ cannot be models of $P\cup F_\I$ because
$F_\I\not\subseteq M[I,J]$, resp.\ $F_\I\not\subseteq M'[I,J]$.
We now analyze these three types of potential candidates:
\begin{itemize}
\item
$\M=U$: Then, for instance, $M[\I,J]\subset U$ is a model of 
$(P\cup F_\I)^\M=P^\M\cup F_\I$ for any $J \subseteq Y$. Thus, $\M\notin\AS(P\cup F_\I)$.
\item
$\M=M[\I,J]$ for some $J\subseteq Y$. Then, by the properties 
of $\Phi$-reductions,
$O[\I]\subset \M$ is a model of $(P\cup F_\I)^\M=P^\M\cup F_\I$.
Thus, $\M\notin\AS(P\cup F_\I)$.
\item 
$\M=M'[\I,J]$ for some $J\subseteq Y$. 
By the initial assumption, there exists a $K_Y\subseteq Z$ with 
$\I\cup J \cup K_Y\not\models \phi$.
Then, by the properties of $\Phi$-reductions,
$N[\I,J,K]\subset \M$ is a model of $P^\M$ and thus also of
$(P\cup F_\I)^\M$.
Hence, $\M\notin\AS(P\cup F_\I)$.
\end{itemize}

In each of the cases we have obtained $\M\notin\AS(P\cup F_\I)$, hence $\AS(P\cup F_\I)=\emptyset$ and $P$ is not super-coherent.
\medskip

\noindent
Suppose that $\Phi$ is true. It is sufficient to show that for each $F\subseteq U$, 
$\AS(P\cup F)\neq \emptyset$. We have the following cases:

If $\{s,\overline{s}\}\subseteq F$ for some $s\in X\cup Y$ or 
$\{u,w\}\subseteq F$. Then $U\in\AS(P\cup F)$ since $U$ is a model 
of $P\cup F$ and each potential model $M\subset U$ of the reduct
$P^U$ (see the properties of $\Phi$-reductions) does not satisfy $F\subseteq M$; thus each such $M$ is not a model of $P^U\cup F=(P\cup F)^U$. 

Otherwise, we have $F\subseteq M[I,J]$ or $F\subseteq M'[I,J]$ for some 
$I\subseteq X$, $J\subseteq Y$. In case $F\subseteq M[I,J]$ and $F\not\subseteq O[I]$, we 
observe that $M[I,J]\in\AS(P\cup F)$ since $O[I]$ is the only model 
(being a proper subset of $M[I,J]$)
of the reduct $P^{M[I,J]}$. Thus for each such $F$ there cannot be a model 
$M\subset M[I,J]$ of $P^{M[I,J]}\cup F=(P\cup F)^{M[I,J]}$. As well, 
in case $F\subseteq M'[I,J]$, where $w\in F$, 
$M'[I,J]$ can be shown to be an answer set of 
$P\cup F$. 
Note that the case $F\subseteq M'[I,J]$ with $w\notin F$ is already
taken care of since for such $F$ we have $F\subseteq M[I,J]$ as well.

It remains to consider the case $F\subseteq O[I]$ for each $I\subseteq X$.
We show that 
$M'[I,J]$ 
is an answer set of $P\cup F$,
for some $J\subseteq Y$. Since $\Phi$ is true, 
we know that, for each $I\subseteq X$, there exists a $J_I\subseteq Y$ 
such that, for all $K\subseteq Z$, $I\cup J_I\cup K\models \phi$. 
As can be verified by the properties of $\Phi$-reductions, then 
there is no model $M\subset M'[I,J_I]$ of $P^{M'[I,J_I]}$. Consequently,
there is also no such model of $(P\cup F)^{M'[I,J_I]}$, and thus 
$M'[I,J_I]\in\AS
(P\cup F)$.

So in each of these cases $\AS(P\cup F) \neq \emptyset$ and since
these cases cover all possible $F \subseteq U$, we obtain that $P$ is
super-coherent.

\medskip

In total we have shown that $\Phi$ being false implies that any
$\Phi$-reduction $P$ is not super-coherent, while $\Phi$ being true
implies that any $\Phi$-reduction is super-coherent, which proves the
lemma.
\end{proof}

It remains to show that for any QBF  
of the desired form, a 
$\Phi$-reduction can be obtained in polynomial time (w.r.t.\ the size of $\Phi$).
For the construction below, let us denote a negated atom $a$ in 
the propositional part of the QBF $\Phi$ as $\overline{a}$.

\begin{definition}\label{def:phi_program}
For any QBF $\Phi = \forall X\exists Y\forall Z\phi$ with 
$\phi=\bigvee_{i=1}^n l_{i,1}\land\cdots\land l_{i,m_i}$ a DNF (i.e., a disjunction of conjunctions over literals), we define
\begin{eqnarray}
\label{eq:red:first}
P_\Phi & = & 
\{  x \lor \overline{x} \la;\;
u \la x,\overline{x};\;
w \la x,\overline{x};\;
x \la u,w;\;
\overline{x} \la u,w\mid x\in X\}\cup{}\\
&&
\{  y \lor \overline{y} \la v;\;
u \la y,\overline{y};\;
w \la y,\overline{y};\;
y \la u,w;\; \nonumber \\
\label{eq:red:second}
&&\hphantom{\{}
\overline{y} \la u,w;\;
v \la y;\;
v \la \overline{y}
\mid y\in Y\}\cup{}\\
&&
\{  z \lor \overline{z} \la v;\;
u \la z,\naf\ w;\;
u \la \overline{z},\naf\ w;\;
v \la z;\;
v \la \overline{z};\; \nonumber \\
&&\hphantom{\{}
\label{eq:red:third}
z \la w;\;
\overline{z} \la w;\;
z \la u;\;
\overline{z} \la u;\;
w \vee u \leftarrow z, \overline{z}
\mid z\in Z\}\cup{}\\
\label{eq:red:fourth}
&& \{ w \vee u \la l_{i,1},\ldots,l_{i,m_i} \mid 1\leq i \leq  n\} \\
\label{eq:red:fifth}
&& \{
  v \la w;\;
  v \la u;\;
  v \la \naf\ u\}.
\end{eqnarray}
\end{definition}

The program above contains atoms associated with literals in $\Phi$ and three auxiliary atoms $u,v,w$.
Intuitively, truth values for variables in $X$ are guessed by rules in (\ref{eq:red:first}), and truth values for variables in $Y \cup Z$ are guessed by rules in (\ref{eq:red:second})--(\ref{eq:red:third}) provided that $v$ is true.
Moreover, rules in (\ref{eq:red:first})--(\ref{eq:red:second}) derive all atoms in $X \cup \overline{X} \cup Y \cup \overline{Y} \cup \{u,w\}$ whenever both $u$ and $w$ are true, or in case an inconsistent assignment for some propositional variable in $X \cup Y$ is forced by the addition of a set of facts to the program.
On the other hand, all atoms associated with variables $Z$ are inferred by rules in (\ref{eq:red:third}) if one of $u$ and $w$ is true.
Atoms $u$ and $w$ can be inferred for instance by some rule in (\ref{eq:red:fourth}) whenever truth values of atoms of $P_\Phi$ represent a satisfying assignment for $\phi$.
Furthermore, rules in (\ref{eq:red:fifth}) are such that every model of $P_\Phi$ contains $v$, and any answer set of $P_\Phi$ must also contain $w$.
We finally observe that $w$ can be derived by some rule in (\ref{eq:red:fourth}) in the program reduct, which is the case if truth values of atoms in $X \cup \overline{X} \cup Y \cup \overline{Y}$ represent a consistent assignment for propositional variables in $X \cup Y$ satisfying $\forall Z \phi$.

Obviously,  the program from above definition can be constructed in polynomial
time in the size of the reduced QBF. To conclude the proof of Theorem~\ref{thm:disj} it is thus sufficient to show the following relation.

\begin{lemma}
For any QBF $\Phi = \forall X\exists Y\forall Z\phi$, the program
$P_\Phi$ is a $\Phi$-reduction.
\end{lemma}
\begin{proof}
Obviously, $\At(P_\Phi)$ contains the atoms as required in 
1) of Definition~\ref{def:red}.
We continue to show 2). 
To see that $U$ is a model of $P_\Phi$ is obvious. 
We next show that 
the remaining models $M$ are all of the form $M[I,J]$ or $M'[I,J]$. 
First we have $v\in M$ because of the rules $v\la u$ and $v\la \naf\ u$ in (\ref{eq:red:fifth}). 
In case $w\in M$, $Z\cup \overline{Z}\subseteq M$ 
by the rules in (\ref{eq:red:third}).
In case $w\notin M$, 
we have  
$K\cup \overline{(Z\setminus K)} \subseteq M$ for some $K\subseteq Z$, since $v\in M$ 
and by (\ref{eq:red:third}). 
But then, since $w\notin M$, $u\in M$ holds (rules $u\la z,\naf\ w$ resp.\
$u\la \overline{z},\naf\ w$).
Hence, also here $Z\cup \overline{Z}\subseteq M$.
In both cases, we observe that by (\ref{eq:red:first}) and (\ref{eq:red:second}), 
$I\cup \overline{(X\setminus I)}\cup J\cup \overline{(Y\setminus J)}\subseteq M$, for some $I\subseteq X$ and $J\subseteq Y$. This yields the desired models, 
$M[I,J]$, $M'[I,J]$. It can be checked that no other model exists by 
showing that for $N\not\subseteq M[I,J]$, resp.\ $N\not\subseteq M'[I,J]$, 
$N=U$ follows.

We next show that, for each $I\subseteq X$ and $J\subseteq Y$, 
$P^{M[I,J]}$ and $P^{M'[I,J]}$ possess the required models.
Let us start by showing that $O[I]$ is a model of $P^{M[I,J]}$. In fact,
it can be observed that all of the rules of the form $x \vee \overline{x} \leftarrow$
in (\ref{eq:red:first})
are satisfied because either $x$ or $\overline{x}$ belong to $O[I]$,
while all of the other rules in $P^{M[I,J]}$ are satisfied because of a false body literal.
We also note that each strict subset of $O[I]$ does not satisfy some rule
of the form $x \vee \overline{x} \leftarrow$, and thus it is not a model of $P^{M[I,J]}$.
Similarly, any interpretation $W$ such that $O[I] \subset W \subset M[I,J]$
does not satisfy some rule in $P^{M[I,J]}$ (note that rules of the form
$u \leftarrow z$ and $u \leftarrow \overline{z}$ occur in $P^{M[I,J]}$
because $w \not\in M[I,J]$; such rules are obtained by rules in (\ref{eq:red:third})).

Let us now consider $P^{M'[I,J]}$ and let $W \subseteq M'[I,J]$ be one of its models.
We shall show that either $W = M'[I,J]$, or $W = N[I,J,K]$ for some 
$K \subseteq Z$ such that $I \cup J \cup K \not\models \phi$.
Note that $v$ is a fact in $P^{M'[I,J]}$, hence $v$ must belong to $W$.
By (\ref{eq:red:first}) and (\ref{eq:red:second}), since $v \in W$ and $W \subseteq M'[I,J]$,
we can conclude that all of the atoms in 
$I\cup \overline{(X\setminus I)}\cup J\cup \overline{(Y\setminus J)}$
belong to $W$. Consider now the atom $w$. If $w$ belongs to $W$, by the rules in
(\ref{eq:red:third}) we conclude that all of the atoms in $Z \cup \overline{Z}$ belong to $W$,
and thus $W = M'[I,J]$. Otherwise, if $w \not\in W$,
by the rules of the form $z \vee \overline{z} \leftarrow v$ in (\ref{eq:red:third}), there must
be a set $K \subseteq Z$ such that $K \cup \overline{(Z\setminus K)}$ is 
contained in $W$. Note that no other atoms in $Z \cup \overline{Z}$ can belong to $W$
because of the last rule in (\ref{eq:red:third}). Hence, $W = N[I,J,K]$.
Moreover, $w \not\in W$ and $u \not\in W$ imply that 
$I \cup J \cup K \not\models \phi$ holds
because of (\ref{eq:red:fourth}).

Finally, one can show that $P^U$ does not yield additional models as 
those which are already present by other models.
Let $W \subseteq U$ be a model of $P^U$.
By (\ref{eq:red:first}), $O[I] \subseteq W$ must hold for some $I \subseteq X$.
Consider now the atom $v$. If $v \not\in W$, we conclude that the model $W$ is
actually $O[I]$. We can thus consider the other case, i.e. $v \in W$.
By (\ref{eq:red:second}), $J \cup \overline{(Y \setminus J)} \subseteq W$ must hold
for some $J \subseteq Y$. 
Consider now the atom $u$. If $u \in W$, we have $Z \cup \overline{Z} \subseteq W$
because of (\ref{eq:red:third}).
If no other atom belongs to $W$, then $W = M[I,J]$ holds.
Otherwise, if any other atom belongs to $W$,
it can be checked that $W$ must be equal to $U$.
We can then consider the case in which $u \not\in W$, and the atom $w$.
Again, we have two possibilities.
If $w$ belongs to $W$, by (\ref{eq:red:third})
we conclude that all of the atoms in $Z \cup \overline{Z}$ belong to $W$,
and thus either $W = M'[I,J]$ or $W = U$. Otherwise, if $w \not\in W$,
by the rules of the form $z \vee \overline{z} \leftarrow v$ in (\ref{eq:red:third}), there must
be a set $K \subseteq Z$ such that $K \cup \overline{(Z\setminus K)}$ is 
contained in $W$. Note that no other atoms in $Z \cup \overline{Z}$ can belong to $W$
because of the last rule in (\ref{eq:red:third}). Hence, $W = N[I,J,K]$.
Moreover, because of (\ref{eq:red:fourth}), $w \not\in W$ and $u \not\in W$ imply that 
$I \cup J \cup K \not\models \phi$ holds.
\end{proof}

Note that the program from Definition~\ref{def:phi_program} does not contain
constraints. As a consequence, the $\PiP{3}$-hardness result presented in this section
also holds if we only consider disjunctive \ASP\ programs without constraints.

\subsection{Proof of Theorem~\ref{thm:norm}}\label{sec:proofnorm}

Membership follows by the straight-forward nondeterministic 
algorithm for the complementary problem presented in the previous section.
We have just to note that a $co-\NP$ oracle can be used for checking the 
consistency of a normal program.
Thus, $\PiP{2}$-membership is established.

\bigskip

For the hardness we reduce the $\PiP{2}$-complete problem of deciding 
whether QBFs of the form 
$\forall X\exists Y\phi$ are true to the problem of super-coherence.
Without loss of generality, we can consider $\phi$ to be in CNF and, indeed, 
$X\neq\emptyset$, and $Y\neq \emptyset$.
We also assume that each clause of $\phi$ contains at least one variable from $X$
and one from $Y$.
More precisely, we shall adapt the notion
of $\Phi$-reduction to normal programs. In particular, we have to take into
account a fundamental difference between disjunctive and normal programs:
while disjunctive programs allow for using disjunctive rules for guessing
a subset of atoms, such a guess can be achieved only by means of unstratified 
negation within a normal program. For example, one atom in a set $\{x,y\}$
can be guessed by means of the following disjunctive rule: $x \vee y \la$.
Within a normal program, the same result can be obtained by means of the
following rules: $x \la \naf\ y$ and $y \la \naf\ x$. However, these last rules
would be deleted in the reduced program associated with a model containing 
both $x$ and $y$, which would allow for an arbitrary subset of $\{x,y\}$
to be part of a model of the reduct. More generally speaking, we have to 
take Property (P2), as introduced in Section~\ref{sec:preliminaries},
into account. This makes the following definition a bit more cumbersome
compared to Definition \ref{def:red}.

\begin{definition}\label{def:rednormrefined:2}
Let $\Phi = \forall X\exists Y\phi$ be a QBF with $\phi$ in CNF.
We call any program $P$ satisfying the following properties 
a \emph{$\Phi$-norm-reduction}:
\begin{enumerate}
\item $P$ is given over atoms $U=X\cup Y\cup 
\overline{X} \cup 
\overline{Y} \cup 
\{v,w\}$, where all atoms in sets 
$\overline{S}=\{\overline{s} \mid s\in S\}$ ($S \in \{X,Y,Z\}$) and $\{v,w\}$ are fresh and mutually disjoint;
\item $P$ has the following models: 
\begin{itemize}
\item 
for each $J \subseteq Y$, and for each $J^*$ such that
$J \cup \overline{(Y \setminus J)} \subseteq J^* \subseteq Y \cup \overline{Y}$
$$
O[J^*] = X \cup \overline{X} \cup J^* \cup \{v,w\};
$$
\item 
for each 
$I\subseteq X$,
$$
M[I]=I\cup \overline{(X\setminus I)}\cup \{v\};
$$
\item 
for each 
$I\subseteq X$, $J\subseteq Y$, such that
$I\cup J\models \phi$,
$$
N[I,J]=
I\cup \overline{(X\setminus I)}\cup 
J\cup \overline{(Y\setminus J)}\cup 
 \{w\};
$$
\end{itemize}
\item  the only models of a reduct 
$P^{M[I]}$ are $M[I]$ and $M[I] \setminus \{v\}$; the only model of a reduct  
$P^{N[I,J]}$ is $N[I,J]$;
\item each model $M$ of a reduct $P^{O[J^*]}$ satisfies the following properties:
 \begin{enumerate}
  \item 
  for each $y \in Y$ such that $y \in O[J^*]$ and $\overline{y} \not\in O[J^*]$,
  if $w \in M$, then $y \in M$;
  \item
  for each $y \in Y$ such that $\overline{y} \in O[J^*]$ and 
  $y \not\in O[J^*]$, if $w \in M$, then $\overline{y} \in M$;
  \item
  if $(Y \cup \overline{Y}) \cap M \neq \emptyset$, then $w \in M$;
  \item
  if there is a clause $l_{i,1} \vee \cdots \vee l_{i,m_i}$ of $\phi$ such that
  $\{\overline{l}_{i,1},\ldots,\overline{l}_{i,m_i}\} \subseteq M$, then
  $v \in M$;
  \item
  if there is an $x \in X$ such that $\{x,\overline{x}\} \subseteq M$,
  or there is a $y \in Y$ such that $\{y,\overline{y}\} \subseteq M$,
  or $\{v,w\} \subseteq M$, then it must hold that 
  $X \cup \overline{X} \cup \{v,w\} \subseteq M$.
 \end{enumerate}
\end{enumerate}
\end{definition}

Similarly as in the previous section,
the models of the reducts $P^{O[J^*]}$ 
as given in Item~4 are not specified for particular purposes, but 
are required to allow for a realization via normal programs taking
the set of models specified in Items~2 and~3 as well as 
properties (P1) and (P2) 
from Section~\ref{sec:preliminaries} into account.

\begin{lemma}
For any QBF $\Phi = \forall X\exists Y\phi$ with $\phi$ in CNF,
a $\Phi$-norm-reduction is super-coherent if and only if 
$\Phi$ is true.
\end{lemma}
\begin{proof}
Suppose that $\Phi$ is false. 
Hence, there exists an $\I\subseteq X$ such that, for all 
$J\subseteq Y$, 
$\I\cup J \not\models \phi$.
Now, let $P$ be any $\Phi$-norm-reduction and 
$F_\I = \I\cup \overline{(X\setminus \I)}$. 
We show that $\AS(P\cup F_\I)=\emptyset$, 
thus $P$ is not super-coherent.
Let $\M$ be a model of $P \cup F_\I$.
Since $P$ is a $\Phi$-norm-reduction, the 
only candidates for $\M$ are 
$O[J^*]$ for some $J \subseteq Y$ and $J^*$ such that
$J \cup \overline{(Y \setminus J)} \subseteq J^* \subseteq Y \cup \overline{Y}$,
$M[\I]$, and $N[\I,J']$, 
where $J'\subseteq Y$ satisfies $\I\cup J'\models \phi$.
However, from our assumption 
(for all
$J\subseteq Y$,
$\I\cup J \not\models \phi$),
no such $N[\I,J']$ exists.
Thus, it remains to consider $O[J^*]$ and $M[\I]$.
By the properties of $\Phi$-norm-reductions, $M[\I] \setminus \{v\}$ is a
model of $P^{M[\I]}$, and hence $M[\I] \setminus \{v\}$ is also a model of
$P^{M[\I]} \cup F_\I = (P \cup F_\I)^{M[\I]}$. Thus, $M[\I]$ is not an answer
set of $P \cup F_\I$.
Due to property (P1),
$M[\I] \setminus \{v\}$ is also a model of $P^{O[J^*]} \cup F_\I = (P \cup F_\I)^{O[J^*]}$, for any $O[J^*]$ and so none of these $O[J^*]$ are answer sets of $P \cup F_\I$ either.
Since this means that no model of $P \cup F_\I$ is an answer set, we conclude $\AS(P \cup F_\I) = \emptyset$, and hence $P$ is not super-coherent.
\medskip

\noindent
Suppose that $\Phi$ is true. It is sufficient to show that, for each $F\subseteq U$, 
$\AS(P\cup F)\neq \emptyset$. We distinguish the following cases for $F\subseteq U$:

$F\subseteq I\cup \overline{(X\setminus I)}\cup \{v\}$ for some $I\subseteq X$:
If $v \in F$, then $M[I]$ is the unique model of $P^{M[I]} \cup F = (P \cup F)^{M[I]}$,
and thus $M[I] \in \AS(P\cup F)$.
Otherwise, if $v \notin F$, since $\Phi$ is true, there exists a $J \subseteq Y$
such that $I\cup J \models \phi$. Thus, $N[I,J]$ is a model of $P\cup F$, and since no subset of $N[I,J]$ is a model of $(P\cup F)^{N[I,J]}$ (by property 3 of $\Phi$-norm-reductions), $N[I,J]\in \AS(P\cup F)$.

$I \cup \overline{(X \setminus I)} \subset F \subseteq N[I,J]$ for some
$I \subseteq X$ and $J\subseteq Y$ such that $I \cup J \models \phi$:
In this case $N[I,J]$ is a model of $P \cup F$ and, by property 3 of
$\Phi$-norm-reductions, $N[I,J]$ is also the unique model of 
$P^{N[I,J]} \cup F = (P \cup F)^{N[I,J]}$.

$I \cup \overline{(X \setminus I)} \subset F \subseteq N[I,J]$ for some
$I \subseteq X$ and $J\subseteq Y$ such that $I \cup J \not\models \phi$:
We shall show that $O[J]$ is an answer set of $P \cup F$ in this case.
Let $M$ be a model of $P^{O[J]} \cup F = (P \cup F)^{O[J]}$.
Since $I \cup \overline{(X \setminus I)} \subset F \subseteq N[I,J]$,
either $w \in F$ or $(Y \cup \overline{Y}) \cap F \neq \emptyset$.
Clearly, $F \subseteq M$ and so $w \in M$ in the first case.
Note that $w \in M$ holds also in the second case because of property 4(c) of $\Phi$-norm-reductions.
Thus, as a consequence of
properties 4(a) and 4(b) of $\Phi$-norm-reductions, $J \cup \overline{(Y \setminus J)} \subseteq M$ holds.
Since $I \cup J \not\models \phi$ and because of property 4(d) of $\Phi$-norm-reductions,
$v \in M$ holds.
Finally, since $\{v,w\} \subseteq M$ and because of property 4(e) of $\Phi$-norm-reductions,
$X \cup \overline{X} \subseteq M$ holds and, thus, $M = O[J]$.

In all other cases, either $\{v,w\} \subseteq F$, or there is an $x \in X$
such that $\{x,\overline{x}\} \subseteq F$, or there is a $y \in Y$ such that
$\{y,\overline{y}\} \subseteq F$.
We shall show that in such cases there is an $O[J^*]$ which is an answer set of
$P \cup F$.
Let $O[J^*]$ be such that $J^* = F \cap (Y \cup \overline{Y})$ and 
let $M$ be a model of $P^{O[J^*]} \cup F = (P \cup F)^{O[J^*]}$
such that $M \subseteq O[J^*]$.
We shall show that $O[J^*] \subseteq M$ holds, which would imply that
$O[J^*] = M$ is an answer set of $P\cup F$.
Clearly, $F \subseteq M$ holds.
By property 4(e) of $\Phi$-norm-reductions, $X \cup \overline{X} \cup \{v,w\} \subseteq M$ holds.
Thus, by property 4(a) of $\Phi$-norm-reductions and because $w \in M$, 
it holds that $y \in M$ for each $y \in Y$ such that 
$y \in O[J^*]$ and $\overline{y} \notin O[J^*]$.
Similarly, by property 4(b) of $\Phi$-norm-reductions and because $w \in M$, 
it holds that $\overline{y} \in M$ for each $y \in Y$ such that 
$\overline{y} \in O[J^*]$ and $y \notin O[J^*]$.
Moreover, for all $y \in Y$ such that $\{y,\overline{y}\} \subseteq O[J^*]$,
it holds that $\{y,\overline{y}\} \subseteq F \subseteq M$.
Therefore, $O[J^*] \subseteq M$ holds and, consequently, $O[J^*] \in \AS(P\cup F)$.

So in each of these cases $\AS(P\cup F) \neq \emptyset$ and since
these cases cover all possible $F \subseteq U$, we obtain that $P$ is
super-coherent.

\medskip

Summarizing, we have shown that $\Phi$ being false implies that any
$\Phi$-norm-reduction $P$ is not super-coherent, while $\Phi$ being true
implies that any $\Phi$-norm-reduction is super-coherent, which proves the
lemma.
\end{proof}

We have now to show that for any QBF of the desired form, a 
$\Phi$-norm-reduction can be obtained in polynomial time
(w.r.t.\ the size of $\Phi$).

\begin{definition}\label{def:norm_phi_program}
For any QBF $\Phi = \forall X\exists Y\phi$ with 
$\phi=\bigwedge_{i=1}^n l_{i,1}\lor\cdots\lor l_{i,m_i}$ in CNF, we define
\begin{eqnarray}
\label{eq:nred:first}
N_\Phi & = & 
\{  
x \la \naf\ \overline{x};\; 
\overline{x} \la \naf\ x\; \mid x \in X\}
\cup{} \\
\label{eq:nred:second}
&&
\{  
y \la \naf\ \overline{y}, w;\; 
\overline{y} \la \naf\ y, w;\;
w \la y;\; w \la \overline{y}\;
\mid y\in Y\}
\cup{} \\
\nonumber
&& \{z \la v,w;\;
z \la x, \overline{x};\;
z \la y, \overline{y}\; \mid z\in X\cup \overline{X} \cup \{v, w\},\\
\label{eq:nred:third}
&& \hspace{17.5em} x \in X, y \in Y\}
\cup{}\\
\label{eq:nred:fourth}
&& \{ v  \la \overline{l}_{i,1},\ldots,\overline{l}_{i,m_i} \mid 1\leq i \leq  n\} \cup{} \\
\label{eq:nred:fifth}
&& \{
  w \la \naf\ v\}.
\end{eqnarray}
\end{definition}

The program above contains atoms associated with literals in $\Phi$ and two auxiliary atoms $v,w$.
Intuitively, $v$ is derived by some rule in (\ref{eq:nred:fourth}) whenever some clause of $\phi$ is violated.
Otherwise, if $v$ is not derived, truth of $w$ is inferred by default because of rule (\ref{eq:nred:fifth}).
Moreover, truth values for variables in $X$ are guessed by rules in (\ref{eq:nred:first}), and truth values for variables in $Y$ are guessed by rules in (\ref{eq:nred:second}) provided that $w$ is true.
Finally, rules in (\ref{eq:nred:third}) derive all atoms in $X \cup \overline{X} \cup \{v,w\}$ whenever both $v$ and $w$ are true, or in case an inconsistent assignment for some propositional variable is forced by the addition of a set of facts to the program.
It turns out that any answer set for such a program is such that truth values of atoms in $X \cup \overline{X} \cup Y \cup \overline{Y}$ represent a consistent assignment for propositional variables in $X \cup Y$ satisfying $\phi$.

Again, note that the program from the above definition can be constructed in polynomial
time in the size of the reduced QBF. To conclude the proof, 
it is thus sufficient to show the following relation.

\begin{lemma}
For any QBF $\Phi = \forall X\exists Y \phi$ with $\phi$ in CNF, 
the program $N_\Phi$ is a $\Phi$-norm-reduction.
\end{lemma}
\begin{proof}
We shall first show that $N_\Phi$ has the requested models.
Let $M$ be a model of $N_\Phi$. Let us consider the atoms $v$ and $w$.
Because of the rule $w \la \naf\ v$ in (\ref{eq:nred:fifth}), three cases are possible:
\begin{enumerate}
\item
$\{v,w\} \subseteq M$.
Thus, $X \cup \overline{X} \subseteq M$ holds because of (\ref{eq:nred:third}).
Moreover, there exists $J \subseteq Y$ such that
$J \cup \overline{(Y \setminus J)} \subseteq M$ because of (\ref{eq:nred:second}).
Note that any other atom in $U$ could belong to $M$. These are the models
$O[J^*]$.

\item
$v \in M$ and $w \notin M$. 
Thus, there exists $I \subseteq X$ such that
$I \cup \overline{(X \setminus I)} \subseteq M$ because of (\ref{eq:nred:first}).
Moreover, no atoms in $Y \cup \overline{Y}$ belong to $M$ because of
(\ref{eq:nred:second}) and $w \notin M$ by assumption.
Thus, $M = M[I]$ in this case.

\item
$v \notin M$ and $w \in M$.
Thus, there exist $I \subseteq X$ and $J \subseteq Y$ such that
$I \cup \overline{(X \setminus I)} \subseteq M$ and
$J \cup \overline{(Y \setminus J)} \subseteq M$ because of (\ref{eq:nred:first})
and (\ref{eq:nred:second}).
Hence, in this case $M = N[I,J]$ and, because of (\ref{eq:nred:fourth}),
it holds that $I \cup J \models \phi$.
\end{enumerate}

Let us consider a reduct $P^{M[I]}$ and one of its models $M \subseteq M[I]$.
First of all, note that $P^{M[I]}$ contains a fact for each atom in 
$I \cup \overline{(X \setminus I)}$. 
Thus, $I \cup \overline{(X \setminus I)} \subseteq M$ holds.
Note also that, since each clause of $\phi$ contains at least one variable from $Y$, 
all of the rules of (\ref{eq:nred:fourth}) have at least one false body literal.
Hence, either $M = M[I]$ or $M = M[I] \setminus \{v\}$, as required by
$\Phi$-norm-reductions.

For a reduct $P^{N[I,J]}$ such that $I \cup J \models \phi$ it is enough to
note that $P^{N[I,J]}$ contains a fact for each atom of $N[I,J]$.

Let us consider a reduct $P^{O[J^*]}$ and one of its models $M \subseteq O[J^*]$.
The first observation is that for each $y \in Y$ such that $y \in O[J^*]$
and $\overline{y} \notin O[J^*]$, the reduct $P^{O[J^*]}$ contains a rule
of the form $y \la w$ (obtained by some rule in (\ref{eq:nred:second})).
Similarly, for each $y \in Y$ such that $\overline{y} \in O[J^*]$
and $y \notin O[J^*]$, the reduct $P^{O[J^*]}$ contains a rule
of the form $\overline{y} \la w$ (obtained by some rule in (\ref{eq:nred:second})).
Hence, $M$ must satisfy properties 4(a) and 4(b) of $\Phi$-norm-reductions.
Property 4(c) is a consequence of (\ref{eq:nred:second}),
property 4(d) follows from (\ref{eq:nred:fourth}) and, finally,
property 4(e) must hold because of (\ref{eq:nred:third}).
\end{proof}

Note that the program from Definition~\ref{def:norm_phi_program} does not contain
constraints. As a consequence, the $\PiP{2}$-hardness result presented in this section
also holds if we only consider normal \ASP\ programs without constraints.

\section{Simulating General Answer Set Programs Using Super-Coherent Programs}
\label{sec:embedding}

Since \ASPSC{} programs are a proper subset of \ASP{} programs, a
natural question to ask is whether the restriction to \ASPSC{}
programs limits the range of problems that can be solved.
In this section we show that this is not the case, i.e.,
all problems solvable in \ASP{} can be encoded in \ASPSC{},
and thus benefit from the optimization potential provided by DMS.
Although these results are interesting from a theoretical point of 
view, we do not suggest that they have to be employed in practice.

Often \ASP{}
is associated with the decision problem of whether a program $P$ has
any answer set (the coherence problem), that is testing whether
$\AS(\P) \neq \emptyset$. Of course, the coherence problem becomes
trivial for \ASPSC{} programs. Another type of decision problem
associated with \ASP{} is query answering, and indeed in this section
we show that using query answering it is possible to simulate both
coherence problems and query answering of \ASP{} by query answering
over \ASPSC{} programs. Using the same construction, we show that also
answer set computation problems for \ASP{} programs can be simulated
by appropriate \ASPSC{} programs. While these constructions always
yield disjunctive programs, we also show how to adapt them in order to
yield normal \ASPSC{} programs when starting from normal programs.

We start by assigning to each disjunctive \ASP\ program a corresponding
super-coherent program. In order to simplify the presentation, and without loss of 
generality, in this section we will only consider programs without constraints.

\begin{definition}\label{def:strat}
Let $\p$ be a program the atoms of which belong to a countable set
$\U$. We construct a new program $\strat{\P}$ by using atoms of the following set:
$$\strat{\U} = \U \cup \{\alpha^T \mid \alpha \in \U\} \cup \{\alpha^F \mid \alpha \in \U\} \cup \{\fail\},$$
where $\alpha^T$, $\alpha^F$ and $\fail$ are fresh symbols not belonging to $\U$.
Program $\strat{\P}$ contains the following rules:
\begin{itemize}
\item
for each rule $\R$ of $\P$, a rule $\strat{\R}$ such that 
\begin{itemize}
\item $\head{\strat{\R}} = \head{\R}$ and 
\item $\body{\strat{\R}} = \posbody{\R} \cup \{\alpha^F \mid \alpha \in \negbody{\R}\};$
\end{itemize}

\item
for each atom $\alpha$ in $\U$, three rules of the form
\begin{eqnarray}
\label{eq:rew:1} \alpha^T \vee \alpha^F & \la \\
\label{eq:rew:2} \alpha^T & \la & \alpha \\
\label{eq:rew:3} \fail & \la & \alpha^T,\ \naf~\alpha.
\end{eqnarray}
\end{itemize}
\end{definition}
Intuitively, every rule $\R$ of $\P$ is replaced by a new rule $\strat{\R}$ in
which new atoms of the form $\alpha^F$ replace negative literals of $\R$. Thus,
our translation must guarantee that an atom $\alpha^F$ is true if and only if 
the associated atom $\alpha$ is false. In fact, this is achieved by means of
rules of the form (\ref{eq:rew:1}), (\ref{eq:rew:2}) and (\ref{eq:rew:3}):
\begin{itemize}
\item
(\ref{eq:rew:1}) guarantees that either $\alpha^T$ or $\alpha^F$ belongs to 
every answer set of $\strat{\P}$;

\item
(\ref{eq:rew:2}) assures that $\alpha^T$ belongs to every model of $\strat{\P}$ 
containing atom $\alpha$; and 

\item
(\ref{eq:rew:3}) derives $\fail$ if $\alpha^T$ is true but $\alpha$ is false,
that is, if $\alpha^T$ is only supported by a rule of the form (\ref{eq:rew:1}).
\end{itemize}
It is not difficult to prove that the program $\strat{\P}$ is super-coherent.

\begin{lemma}\label{lem:strat-super}
Let $\P$ be a disjunctive program.
Program $\strat{\P}$ is stratified and thus super-coherent.
\end{lemma}
\begin{proof}
All negative literals in $\strat{\P}$ are those in rules of the form
(\ref{eq:rew:3}), the head of which is $\fail$, an atom not occurring elsewhere
in $\P$.
\end{proof}

Proving correspondence between answer sets of $\P$ and $\strat{\P}$ is
slightly more difficult. To this aim, we first introduce some
properties of the interpretations of $\strat{\P}$.

\begin{lemma}\label{lem:strat_relation}
Let $I$ be an interpretation for $\strat{\P}$ such that:
\begin{enumerate}
\item for every $\alpha \in \U$, precisely one of $\alpha^T$ and $\alpha^F$ belongs to $I$;
\item for every $\alpha \in \U$, $\alpha \in I$ if and only if $\alpha^T \in I$.
\end{enumerate}
Then,
for every rule $\R$ of $\P$, the following relation is established:
\begin{eqnarray}
\label{eq:strat_relation}
\posbody{\strat{\R}} \subseteq I &
\Longleftrightarrow & 
\posbody{\R} \subseteq (I \cap \U) \mathit{\ and\ } \negbody{\R} \cap (I \cap \U) = \emptyset.
\end{eqnarray}
\end{lemma}
\begin{proof}
By combining properties of $I$ (item 1 and 2), we have that $\alpha^F \in I$
if and only if $\alpha \not\in I$. The claim thus follows by construction of 
$\strat{\P}$. In fact,  
$\posbody{\R} = \posbody{\strat{\R}} \cap \U$
and
$\negbody{\R} = \{\alpha \mid \alpha^F \in \posbody{\strat{\R}} \setminus \U\}$.
\end{proof}

We are now ready to formalize and prove the correspondence between
answer sets of $\P$ and $\strat{\P}$.

\begin{theorem}\label{thm:strat-equiv}
Let $\P$ be a program and $\strat{\P}$ the program obtained as described in
Definition~\ref{def:strat}. The following relation holds:
$$\AS(\P) = \{M \cap \U \mid M \in \AS(\strat{\P}) \wedge \fail \not\in M\}.$$
\end{theorem}
\begin{proof}
$(\supseteq)$
Let $M$ be an answer set of $\strat{\P}$ such that $\fail \not\in M$.
We shall show that $M \cap \U$ is an answer set of $\P$.

We start by observing that $M$ has the properties required by 
Lemma~\ref{lem:strat_relation}:
\begin{itemize}
\item
The first property is guaranteed by rules of the form (\ref{eq:rew:1}) and 
because atoms of the form $\alpha^F$ occur as head atoms only in these rules;

\item
the second property is ensured by rules of the form (\ref{eq:rew:2}) and 
(\ref{eq:rew:3}), combined with the assumption $\fail \not\in M$.
\end{itemize}
Therefore, relation (\ref{eq:strat_relation}) holds for $M$,
which combined with the assumption that $M$ is a model of $\strat{\P}$, 
implies that $M \cap \U$ is a model of $\P$. We next show that $M \cap \U$ is
also a minimal model of the reduct $\P^{(M \cap \U)}$.

Let $(M \cap \U) \setminus \Delta$ be a model of $\P^{(M \cap \U)}$, for some set
$\Delta \subseteq \U$. We next prove that $M \setminus \Delta$ is a model of $\strat{\P}^M$, which implies that $\Delta = \emptyset$ since $M$ is an answer set of $\strat{\P}$.
All of the rules of $\strat{\P}^M$ obtained from
(\ref{eq:rew:1}), (\ref{eq:rew:2}) and (\ref{eq:rew:3}) are satisfied by $M \setminus \Delta$: Rules of $\strat{\P}^M$ obtained from
(\ref{eq:rew:1}), (\ref{eq:rew:2}) remain equal and since $\Delta \subseteq \U$  and since their heads are not in $\U$, satisfaction by $M$ implies satisfaction by $M \setminus \Delta$. Rules of $\strat{\P}^M$ obtained from
(\ref{eq:rew:3}) are such that $\alpha \not\in M$, in addition since $\fail \not\in M$ also $\fail \not\in M \setminus \Delta$, and since the rule is satisfied by $M$, $\alpha^T \not\in M$ and thus also $\alpha^T \not\in M \setminus \Delta$.
Every remaining rule $\strat{\R} \in \strat{\P}$ is such that $\R \in \P$.
If $\posbody{\strat{\R}} \subseteq M \setminus \Delta \subseteq M$, we can apply
(\ref{eq:strat_relation}) and conclude $\negbody{\R} \cap (M \cap \U) = \emptyset$, that
is, a rule obtained from $\R$ by removing negative literals belongs to $\P^{(M \cap \U)}$.
Moreover, we can conclude that
$\posbody{\R} = \posbody{\strat{\R}} \cap \U \subseteq (M \cap \U) \setminus \Delta$,
and thus $\emptyset \neq \head{\R} \cap (M \cap \U) \setminus \Delta 
= \head{\strat{\R}} \cap M \setminus \Delta$.

\medskip

\noindent
$(\subseteq)$
Let $M$ be an answer set of $\P$.
We shall show that 
$$M' = M \cup \{\alpha^T \mid \alpha \in M\} \cup 
\{\alpha^F \mid \alpha \in \U \setminus M\}$$
is an answer set of $\strat{\P}$.

We first observe that relation (\ref{eq:strat_relation}) holds for $M'$,
which combined with the assumption that $M$ is a model of $\P$ implies that $M'$
is a model of $\strat{\P}$. We can
then show that $M'$ is also a minimal model for the reduct $\strat{\P}^{M'}$.
In fact, we can show that every $N \subseteq M'$ modeling $\strat{\P}^{M'}$
is such that $N = M'$ in two steps:
\begin{enumerate}
\item
$N \setminus \U = M' \setminus \U$. This follows by rules of the form 
(\ref{eq:rew:1}) and by construction of $M'$.
In fact, rules of the form (\ref{eq:rew:1}) belongs to $\strat{\P}^{M'}$, of which $N$ is a model by assumption.
For each $\alpha \in \U$, these rules enforce the presence of at least one of $\alpha^T$ and $\alpha^F$ in $N$.
By construction, $M'$ contains exactly one of $\alpha^T$ or $\alpha^F$ for each $\alpha \in \U$, and we thus conclude $N \setminus \U = M' \setminus \U$.

\item
$N \cap \U = M' \cap \U$.
Note that $M' \cap \U = M$. Moreover, from the assumption $N \subseteq M'$,
we have $N \cap \U \subseteq M' \cap \U = M$. Thus, it is enough to show that 
$N \cap \U$ is a model of $\P^M$ because in this case $N \cap \U = M$
would be a consequence of the assumption $M \in \AS(\P)$.

In order to show that $N \cap \U$ is a model of $\P^M$, let us consider a rule $\R'$
of $\P^M$ with $\posbody{\R'} \subseteq N \cap \U$. Such a rule has been obtained 
from a rule $\R$ of $\P$ such that $\negbody{\R} \cap M = \emptyset$.
Let us consider $\strat{\R} \in \strat{\P}$, recall that $\strat{\R}$ has been obtained from $\R$ by replacing negatively occurring atoms $\alpha$ by $\alpha^F$.
Clearly, $\strat{\R} \in \strat{\P}^{M'}$ because $\negbody{\strat{\R}} = \emptyset$
by construction.
Moreover, since $\posbody{\R'} \subseteq N \cap \U$, we get
$\posbody{\R} \subseteq N \cap \U$ (recall that $\R'$ is the reduct of
$\R$ with respect to $M$, thus $\posbody{\R}=\posbody{\R'}$) and
consequently $\posbody{\strat{\R}} \cap \U \subseteq N \cap \U$ (since
$\posbody{\strat{\R}} \cap \U = \posbody{\R}$). Furthermore,
$\posbody{\strat{\R}} \setminus \U = \{\alpha^F \mid \alpha \in
\negbody{\R}\}$, and since $\negbody{\R}\cap M =\emptyset$ we get
$\posbody{\strat{\R}} \setminus \U \subseteq M' \setminus \U$ (since
$\alpha^F \in M'$ if and only if $\alpha \not\in M$); given $N
\setminus \U = M' \setminus \U$ we also have $\posbody{\strat{\R}}
\setminus \U \subseteq N \setminus \U$. 
In total, we obtain $\posbody{\strat{\R}} \subseteq N$.
Since by assumption $N$ is a model of $\strat{\P}^{M'}$, $\head{\strat{\R}} \cap N \neq \emptyset$ and we can conclude that
$\head{\R'} \cap (N \cap \U) \neq \emptyset$, that is, rule $\R'$ is satisfied
by $N \cap \U$.
\end{enumerate}
Summarizing, no $N \subsetneq M'$ is a model of $\strat{\P}^{M'}$, and hence $M'$ is a minimal model of $\strat{\P}^{M'}$ and thus an answer set of $\strat{\P}$.
\end{proof}

Let us now consider normal programs. Our aim is to define a translation for
associating every normal program with a super-coherent normal program.
Definition~\ref{def:strat} alone is not suitable for this purpose, as rules of
the form (\ref{eq:rew:1}) are disjunctive. However, it is not difficult to prove
that the application of Definition~\ref{def:strat} to normal programs yields
head-cycle free programs.

\begin{lemma}\label{lem:normal-hcf}
If $\P$ is a normal program, $\strat{\P}$ is head-cycle free.
\end{lemma}
\begin{proof}
Since $\P$ is normal, all disjunctive rules in $\strat{P}$ are of
the form (\ref{eq:rew:1}). Atoms $\alpha^T$ and $\alpha^F$ are not involved
in any cycle of dependencies because $\alpha^F$ do not appear in 
any other rule heads.
\end{proof}

Hence, for a normal program $\P$, we construct a head-cycle free
program $\strat{\P}$. It is well known in the literature that a head-cycle free
program $\P$ can be associated to a \emph{uniformly equivalent} normal program 
$\P^\rightarrow$, meaning that $\P \cup F$ and $\P^\rightarrow \cup F$ are
equivalent, for each set of facts $F$. Below, we recall this result.

\begin{definition}[Definition 5.11 of \citeNP{Eiter07}]\label{def:shift}
Let $\P$ be a disjunctive program. We construct a new program $\P^\rightarrow$
as follows:
\begin{itemize}
 \item
 all the rules $\R \in \P$ with $\head{\R} = \emptyset$ belong to $\P^\rightarrow$;

 \item
 for each rule $\R \in \P$ with $\head{\R} \neq \emptyset$, and for each atom $\alpha \in \head{\R}$,
 program $\P^\rightarrow$ contains a rule $\R^\rightarrow$ such that
 $\head{\R^\rightarrow} = \{\alpha\}$, $\posbody{\R^\rightarrow} = \posbody{\R}$
 and $\negbody{\R^\rightarrow} = \negbody{\R} \cup (\head{\R} \setminus \{\alpha\})$.
\end{itemize}
\end{definition}
\begin{theorem}[Adapted from Theorem 5.12 of \citeNP{Eiter07}]\label{thm:shift-uniform}
For any head-cycle free program $\P$, and any set of atoms $F$, it holds that
$\AS(\P \cup F) = \AS(\P^\rightarrow \cup F)$.
\end{theorem}

Thus, by combining Definitions~\ref{def:strat} and \ref{def:shift}, we can associate
every normal program with a super-coherent normal program.

\begin{theorem}
For a program $\P$, let $\strat{\P}^\rightarrow$ be the program obtained by 
applying the transformation in Definition~\ref{def:shift} to the program $\strat{\P}$.
If $\P$ is a normal program, $\strat{\P}^\rightarrow$ is super-coherent and such that:
$$\AS(\P) = \{M \cap \U \mid M \in \AS(\strat{\P}^\rightarrow) \wedge \fail \not\in M\}.$$
\end{theorem}
\begin{proof}
By Lemma~\ref{lem:normal-hcf} and
Theorem~\ref{thm:shift-uniform},
$\AS(\strat{\P} \cup F) = \AS(\strat{\P}^\rightarrow \cup F)$, for any set $F$ of facts. 
Moreover,
with Lemma~\ref{lem:strat-super},
we obtain that 
$\strat{\P}^\rightarrow$ 
is super-coherent.
Finally, by using Theorem~\ref{thm:strat-equiv}, 
$\AS(\P) = \{M \cap \U \mid M \in \AS(\strat{\P}^\rightarrow) \wedge \fail \not\in M\}$
holds.
%
%
\end{proof}

Note that the program $\strat{\P}^\rightarrow$ can be obtained from $\P$ by
applying the rewriting introduced by Definition~\ref{def:strat}, in which rules
of the form (\ref{eq:rew:1}) are replaced by rules of the following form:
\begin{equation}\label{eq:rew:1shift}
\alpha^T \la \naf~\alpha^F \qquad
\alpha^F \la \naf~\alpha^T.  
\end{equation}
We also note that it is possible to use rules (\ref{eq:rew:1shift})
instead of rules (\ref{eq:rew:1}) for the general, disjunctive case,
such that Theorem~\ref{thm:strat-equiv} would still hold. However, the
proof is somewhat more involved.

We can now relate the coherence and query answering problems for
\ASP{} to corresponding query answering problems for \ASPSC{}, and
thus conclude that all of the problems solvable in this way using
\ASP{} are also solvable using \ASPSC{}.

\begin{theorem}
Given a program $\P$ over $\U$,
\begin{enumerate}
\item $\AS(\P) = \emptyset$ if and only if $\fail$ is cautiously true for $\strat{\P}$;
\item a query $\q$ is bravely true for $\P$ if and only if $\qnew$ is bravely true for the \ASPSC{} program $\strat{\P} \cup \{ \qnew \la \q, \naf~\fail \}$;
\item a query $\q$ is cautiously true for $\P$ if and only if $\qnew$ is cautiously true for the \ASPSC{} program $\strat{\P} \cup \{ \qnew \la \q ; \qnew \la \fail \}$,
\end{enumerate}
where $\qnew$ is a fresh atom, which does not occur in $\P$ and
$\strat{\P}$.
\end{theorem}
\begin{proof}
We first observe that programs $\strat{\P} \cup \{ \qnew \la \q, \naf~\fail \}$ and $\strat{\P} \cup \{ \qnew \la \q ; \qnew \la \fail \}$ are \ASPSC{}. In fact, rules $\qnew \la \q, \naf~\fail$ and $\{ \qnew \la \q ; \qnew \la \fail \}$ do not introduce cycles in these programs as $\qnew$ does not occur in $\P$ and $\strat{\P}$. We now prove the other statements of the theorem.

(1) If $\AS(\P) = \emptyset$, by Theorem~\ref{thm:strat-equiv} either
  $\AS(\strat{\P})=\emptyset$ (this will not occur because
  $\strat{\P}$ is super-coherent) or $\fail \in M$ for all $M \in
  \AS(\strat{\P})$. In either case $\fail$ is cautiously true for
  $\strat{\P}$. If $\fail$ is cautiously true for $\strat{\P}$, then
  $\fail \in M$ for all $M \in \AS(\strat{\P})$, hence by
  Theorem~\ref{thm:strat-equiv}, $\AS(\P) = \emptyset$.

(2) Let $\P^+$ denote $\strat{\P} \cup \{ \qnew \la \q, \naf~\fail
  \}$. If $\q$ is bravely true for $\P$, there is $M \in \AS(\P)$ such
  that $\q \in M$, and by Theorem~\ref{thm:strat-equiv} there is an
  $M' \in \AS(\strat{\P})$ such that $\fail \not \in M'$ and $M =
  M'\cap \U$, and hence $\q \in M'$. Therefore $M' \cup \{\qnew \} \in
  \AS(\P^+)$ and thus $\qnew$ is bravely true for $\P^+$. If $\qnew$
  is bravely true for $\P^+$, then there exists one $N \in \AS(\P^+)$
  such that $\qnew \in N$. Then $N^- = N \setminus \{\qnew\}$ is in
  $\AS(\strat{\P})$ and $\fail$ is not in $N$ and $N^-$, while $\q$ is
  in both $N$ and $N^-$. Therefore, by Theorem~\ref{thm:strat-equiv},
  $N'= N\cap \U$ is in $\AS(\P)$ and $\q \in N'$, hence $\q$ is
  bravely true for $\P$.

(3) Let $\P^+$ denote $\strat{\P} \cup \{ \qnew \la \q ; \qnew \la
  \fail \}$. If $\q$ is cautiously true for $\P$, then $\q \in M$ for
  all $M \in \AS(\P)$. By Theorem~\ref{thm:strat-equiv}, each $M' \in
  \AS(\strat{\P})$ either is of the form $M = M'\cap \U$ for some $M
  \in \AS(\P)$ and thus $\q \in M'$ or $\fail \in M'$. In either case
  we get $M' \cup \{\qnew\} \in \AS(\P^+)$ and hence that $\qnew$ is
  cautiously true for $\P^+$. If $\qnew$ is cautiously true for
  $\P^+$, then $\qnew \in N$ for each $N\in \AS(P^+)$. Each of these
  $N$ contains either (a) $\fail$ or (b) $\q$, and $\AS(\strat{P}) =
  \{N \setminus \{\qnew\} \mid N \in \AS(P^+)\}$. By
  Theorem~\ref{thm:strat-equiv}, each $N^-\in \AS(\P)$ is of the form
  $N^-=N\cap \U$ for those $N \in \AS(P^+)$ which do not contain
  $\fail$, hence are of type (b) and contain $\q$. Therefore $\q$ is
  in all $N^-\in \AS(\P)$ and is therefore cautiously true for $\P$.
\end{proof}

\section{Some Implications}\label{sec:impl}

\citeText{Oetsch07b} studied the following problem under 
the name ``uniform equivalence with projection:''
\begin{itemize}
\item
Given two programs $P$ and $Q$, and two sets $A,B$ of atoms, 
$P\equiv^A_B Q$ if and only if for each set $F\subseteq A$ of facts,
$\{ I\cap B \mid I \in \AS(P\cup F)\} =
\{ I\cap B \mid I \in \AS(Q\cup F)\}$. 
\end{itemize}
Let us call $A$ the context alphabet and $B$ the projection alphabet.
As is easily verified the following 
relation holds.

\begin{proposition}\label{prop:relation}
A program $P$ over atoms $U$ is super-coherent if and only if 
$P\equiv^U_\emptyset Q$, where $Q$ is an arbitrary definite Horn program.
\end{proposition}

Note that $P\equiv^U_\emptyset Q$ means $\{ I\cap \emptyset \mid I \in \AS(P\cup F)\} = \{ I\cap \emptyset \mid I \in \AS(Q\cup F)\}$ for all sets $F \subseteq U$. Now observe that for any $F \subseteq U$, both of these sets are either empty or contain the empty set, depending on whether the programs (extended by $F$) have answer sets.
Formally, we have 
\[
\{ I\cap \emptyset \mid I \in \AS(P\cup F)\} = \left\{ 
\begin{array}{cl}
\emptyset & \mbox{ iff } \AS(P\cup F) = \emptyset\\
\{\emptyset\} & \mbox{ iff } \AS(P\cup F) \neq \emptyset
\end{array}
\right.
\]
\[
\{ I\cap \emptyset \mid I \in \AS(Q\cup F)\} = \left\{ 
\begin{array}{cl}
\emptyset & \mbox{ iff } \AS(Q\cup F) = \emptyset\\
\{\emptyset\} & \mbox{ iff } \AS(Q\cup F) \neq \emptyset
\end{array}
\right.
\]

If $Q$ is a definite Horn program, then $\AS(Q\cup F) \neq \emptyset$ for all $F \subseteq U$, and therefore the statement of Proposition~\ref{prop:relation} becomes equivalent to checking whether $\AS(P\cup F) \neq \emptyset$ for all $F \subseteq U$, and thus whether $P$ is super-coherent.

\citeText{Oetsch07b} also investigated the complexity of the problem
of deciding uniform equivalence with projection, reporting
$\PiP{3}$-completeness for disjunctive programs and 
$\PiP{2}$-completeness for normal programs.  However, these hardness
results use bound context alphabets $A\subset U$ (where $U$ 
are all atoms from the compared programs).
Our results thus strengthen the observations of \citeText{Oetsch07b}.
Using Proposition~\ref{prop:relation} and the main results in this paper, we 
can state the following result.
\begin{theorem}
The problem of deciding 
$P\equiv^U_\emptyset Q$ 
for given disjunctive (resp.\ normal) programs $P$ and $Q$ 
is $\PiP{3}$-complete (resp.\ $\PiP{2}$-complete) even in case $U$ is the set of all atoms
occurring in $P$ or $Q$.  Hardness already holds if one the programs is the empty program.
\end{theorem}


\section{Conclusion}\label{sec:concl}

Many recent advances in ASP rely on the adaptions of technologies from
other areas. One important example is the Magic Set method, which
stems from the area of databases and is used in state-of-the-art ASP
grounders. Recent work showed that the ASP variant of this technique
only applies to the class of programs called super-coherent~%
\cite{alvi-fabe-2011-aicom}. Super-coherent programs are those which
possess at least one answer set, no matter which set of facts is added
to them. We believe that this class of programs is interesting per se
(for instance, since there is a strong relation to some problems in
equivalence checking), and also showed that all of the interesting ASP
tasks can be solved using super-coherent programs only.  For these
reasons we have studied the exact complexity of recognizing the
property of super-coherence for disjunctive and normal programs.  Our
results show that the problems are surprisingly hard, viz.\ complete
for $\PiP{3}$ and resp.\ $\PiP{2}$.  Our results also imply that any
reasoning tasks over \ASP{} can be transformed into tasks over
\ASPSC{}. In particular, this means that all query answering tasks
over \ASP{} can be transformed into query answering over \ASPSC{}, on
which the magic set technique can therefore be applied. However, we
believe that the magic set technique will often not produce efficient
rewritings for programs obtained by the automatic transformation of
Section~\ref{sec:embedding}. A careful analysis of this aspect is
however left for future work, as it is also not central to the topics
of this article.

\bibliographystyle{acmtrans}
\bibliography{supercoherent}

\end{document}